\documentclass[11pt,twoside]{article}
\usepackage{amssymb,amsmath,amsthm}
\usepackage{upgreek}
\usepackage{fullpage}
\newtheorem{lemma}{Lemma}
\newcommand{\F}{\mathbb{F}}
\renewcommand{\O}{\mathcal{O}}

\title{A remark on MAKE -- a Matrix Action Key Exchange}
\author{Chris Monico\\
\small{Texas Tech University}\\
\small{c.monico@ttu.edu}
 \and Ayan Mahalanobis\\
 \small{IISER Pune}\\\small{ayanm@iiserpune.ac.in}}
 \date{}
\begin{document}
\maketitle

\section{Introduction}
Rahman and Shpilrain~\cite{sh} proposed a new key-exchange protocol MAKE based on external semidirect product of groups. The purpose of this paper is to show that the key exchange protocol is insecure. We were able to break their challenge problem in under a second.
\section{Description of MAKE} 
Let $\mathcal{G}$ and $\mathcal{H}$ be two
semigroups of $k\times k$ matrices over $\F_p$. 
The semigroup $\mathcal{G}$ is defined 
additively and $\mathcal{H}$ is defined multiplicatively. We define a 
semidirect $\mathcal{G}\ltimes\mathcal{H}$ such that 
$(G_1,H_1)\cdot(G_2,H_2)=(H_2G_1H_2 + G_2,H_1H_2)$, where 
$G_i\in\mathcal{G}$ and $H_i\in\mathcal{H}$ for $i=1,2$. 
For more on the description of the semidirect product used to make MAKE see~\cite{sh}.

Like all key exchange protocols, the purpose of MAKE is for Alice and Bob to set up 
a common key for secure communications over an insecure channel. In the case of MAKE
this was achieved in the following way:

\begin{enumerate}
  \item[(i)] Alice and Bob decide over an insecure channel that they are going to 
    use two invertible matrices $M$ and $H$ over $\F_p$ for some suitable 
    prime $p$. Here $M\in\mathcal{G}$ and $H\in\mathcal{H}$. 
  \item[(ii)] Alice chooses an integer $m$ and Bob an integer $n$. 
    These integers are private information.
  \item[(iii)] Alice computes $(M,H)^m=(A,H^m)$. 
    The operation is the product in the semidirect product defined above. 
    She sends $A$ to Bob but keeps $H^m$ secret.
  \item[(iv)] Bob computes $(M,H)^n=(B,H^n)$ and sends $B$ to Alice and keeps $H^n$ secret.
  \item[(v)] Alice on receiving $B$ computes the first component
    of $(B,Q)\cdot(A,H^m)=(H^mBH^m+A,QH^m)$. 
    The common key is $H^mBH^m+A$. Note that 
    $Q$ is neither known nor needed in this case and the key exchange is successful 
    without an explicit knowledge of $Q$.
  \item[(vi)] Bob on receiving $A$ computes the first component 
    of $(A,Q)\cdot(B,H^n)=(H^nAH^n+B,QH^n)$. 
    The key is $H^nAH^n+B$. Note as before $Q$ is not known.
\end{enumerate}

To check that this key exchange is successful one has to check if
\[H^mBH^m + A = H^nAH^n + B\] is true.
This follows from $A = H^{m-1}MH^{m-1} + H^{m-2}MH^{m-2} + \ldots + HMH + M$ and 
$B = H^{n-1}MH^{n-1} + H^{n-2} M H^{n-2} + \ldots + HMH + M$, and the common key is 
\[
  K = H^{m+n-1} M H^{m+n-1} + H^{m+n-2}MH^{m+n-2} + \ldots + HMH + M.
\]

\section{An attack on MAKE}
  Our attack is based on the following lemma, illustrating that
  recovery of the private parameters is not necessary for an eavesdropper
  to obtain the shared secret key.

\begin{lemma} \label{lem:RS}
  Let $M,H,A,B,m,n$ be as above.
  Suppose $R$ and $S$ are matrices which commute with $H$ and satisfy
  \[
    RMS = HAH + M  - A.
  \]
  Then Alice and Bob's shared secret key $K$ satisfies
  $RBS+A=K$.
\end{lemma}
\begin{proof}
  We simply compute
  \begin{eqnarray*}
    RBS + A 
      &=& R\left( \sum_{i=0}^{n-1} H^iMH^i\right)S + A \\
      &=& \sum_{i=0}^{n-1} H^iRMSH^i + A \\
      &=& \sum_{i=0}^{n-1} H^i(HAH+M-A)H^i + A \\
      &=& \sum_{i=1}^{n} H^iAH^i +\sum_{i=0}^{n-1}H^iMH^i -\sum_{i=0}^{n-1}H^iAH^i + A \\
      &=& H^nAH^n +\sum_{i=0}^{n-1}H^iMH^i 
      = \sum_{i=0}^{m+n-1}H^iMH^i = K.
  \end{eqnarray*}
\end{proof}

Our approach to finding such matrices $R,S$ is as follows.
We will find polynomials $f,g\in\F_p[t]$ for which
\begin{equation} \label{eq:polyeq}
  f(H)M = \left( HAH + M - A\right)g(H),
\end{equation}
with $g(H)$ invertible. If the multiplicative order of $H$ is $\nu$,
then $f(t)=t^m$ and $g(t)=t^{\nu-m}$ satisfy \eqref{eq:polyeq},
so such polynomials necessarily exist. In that case, it follows
that $R=f(H)$ and $S=g(H)^{-1}$ satisfy the hypothesis of Lemma \ref{lem:RS}.

To find $f,g$ satisfying \eqref{eq:polyeq}, first note that by the
Cayley-Hamilton Theorem, we may assume they each have degree at most $k-1$.
Set $Z=HAH+M-A$, and is then sufficient to solve
\begin{eqnarray} \label{eq:linsys}
  f_0M + f_1HM + \cdots + f_{k-1}H^{k-1}M = g_0Z + g_1ZH + \cdots + g_{k-1}ZH^{k-1}.
\end{eqnarray}
This is a homogeneous linear system of $k^2$ equations in the $2k$ unknowns
$f_0,\ldots,f_{k-1},g_0,\ldots, g_{k-1}$. 
We use standard techniques to produce a basis
for the subspace of solutions in $\F_p^{2k}$.
Then simply choose nonzero solutions at random until one is found for which
the resulting $g(H)=g_0I + g_1H + \cdots + g_{k-1}H^{k-1}$ is
invertible. 
We implemented this attack based on the Python code provided by the authors of
\cite{sh} for the case of $k=3$. After performing 200 experiments for each
of several different primes ranging from $p=17$ to their 2000-bit prime number,
we did not encounter a single case where more than one random choice was necessary.
It's not hard to show that if the dimension of the solution space is 1,
then every nonzero solution yields an invertible $g$; but we do not know
if that dimension is always 1.

\section{Analysis of the attack}
The attack presented here has just four steps:
\begin{enumerate}
  \item Compute $Z$.
  \item Find a basis for the solutions to \eqref{eq:linsys},
  \item Choose a nonzero solution to \eqref{eq:linsys} and determine $R=f(H)$
        and $S=g(H)^{-1}$.
  \item Compute $K=RBS + A$.
\end{enumerate}

  Finding $Z$ requires 2 matrix multiplications and 3 additions,
  for a total of $\O(k^3)$ arithmetic operations in $\F_p$.

  To compute a basis for the solutions to \eqref{eq:linsys},
  one explicitly computes $M, HM, H^2M, \ldots, H^{k-1}M$ using $k-1$ 
  matrix multiplications, and then $Z, ZH, ZH^2,\ldots, ZH^{k-1}$
  using another $k-1$ matrix multiplications, for a total of
  $\O(k^4)$ arithmetic operations in $\F_p$. These $2k$ matrices
  are `flattened' into row vectors to form the $2k\times(k^2+2k)$ matrix 
  consisting of the $2k\times k^2$ left submatrix formed by those row vectors,
  augmented with the $2k\times 2k$ identity matrix.
  This matrix is row-reduced, using $\O(4k^2(k^2+2k)) = \O(k^4)$
  arithmetic operations in $\F_p$.
  A basis for the kernel is then found on the right-hand side of the reduced
  matrix, as the row-vectors whose corresponding left-hand sides are zero.
   Computing and inverting $g(H)$
  requires $\O(k^3)$ more arithmetic operations in $\F_p$.
  Supposing that this must be repeated $t$ times before obtaining
  an invertible $g$, we have a total of $\O(tk^4)$ arithmetic operations
  in $\F_p$, which dominates the time required to compute $K$, so
  this is the total number of $\F_p$ arithmetic operations required,
  with the most expensive operation being inversion in $\F_p$.
  Therefore the process requires no more than $\O(tk^4log^3{p})$ bit operations.
  
  We performed 200 experiments with $k=3$ for each
  of the four primes $p=17,19,135257, p_{2000}$, where $p_{2000}$ is the 2000-bit
  prime used by the authors of \cite{sh} in their sample code.
  We did not encounter a case in which $t>1$.
  So we experimentally conclude that the entire attack requires
  $\O(k^4\log^3{p})$ bit operations, which is less than the cube
  of the input size, and hence polynomial-time.
  On a single core of an i7-5557U 3.10 GHz processor, the time required to
  solve the largest $p_{2000}$ instances was about 0.03 seconds.


\bibliography{bibfile}
\end{document}